\documentclass[a4paper]{article}

\usepackage[utf8]{inputenc}
\usepackage{amsmath}
\usepackage{amsthm}

\usepackage{shuffle}
\usepackage{bm}
\usepackage{amssymb} 
\usepackage{mathtools}
\usepackage{url}
\usepackage{complexity}
\usepackage[textsize=footnotesize]{todonotes}
\usepackage{tikz}
\usetikzlibrary{arrows.meta,positioning,automata,shadows,shapes}
\usepackage{cleveref}

\newtheorem{theorem}{Theorem}[section]

\newtheorem{definition}[theorem]{Definition}

\newtheorem{lemma}[theorem]{Lemma}
\newtheorem{proposition}[theorem]{Proposition}
\newtheorem{remark}[theorem]{Remark}
\def\PT{\ComplexityFont{PT}}
\newcommand{\concatdot}{\mathord{\cdot}} 
\newcommand{\Sol}{\mathit{Sol}}
\newcommand{\Nat}{\mathbb{N}}
\newcommand{\tta}{\texttt{a}}
\newcommand{\ttb}{\texttt{b}}
\newcommand{\ttc}{\texttt{c}}
\newcommand{\calA}{\mathcal{A}}
\newcommand{\calR}{\mathcal{R}}
\newcommand{\calL}{\mathcal{L}}
\newcommand{\calJ}{\mathcal{J}}
\newcommand{\calB}{\mathcal{B}}
\newcommand{\calP}{\mathcal{P}}
\newcommand{\bmi}{\bm{i}}
\newcommand{\frakF}{\mathfrak{F}}
\newcommand{\bmF}{\bm{F}}
\newcommand{\bmQ}{\bm{Q}}
\newcommand{\Pcdot}{\bullet}
\newcommand{\Bcdot}{\circ}
\newcommand{\subword}{\preccurlyeq}

\newcommand{\mirror}[1]{{#1}^{\leftarrow}}
\newcommand{\depth}{{\textit{dp}}}
\renewcommand{\setminus}{\smallsetminus}
\newcommand{\step}[1]{{\xrightarrow{\!#1\!}}}
\newcommand{\eqdef}{\stackrel{\mbox{\begin{scriptsize}def\end{scriptsize}}}{=}}
\newcommand{\obracew}[2]{{\overset{#2}{\overbrace{#1}}}}

\author{S.\ Halfon and Ph.\ Schnoebelen}
\title{On shuffle products, acyclic automata and piecewise-testable languages}

\begin{document}
\maketitle

\begin{abstract}
We show that the shuffle $L\shuffle F$ of a piecewise-testable
language $L$ and a finite language $F$ is piecewise-testable. The
proof relies on a classic but little-used automata-theoretic
characterization of piecewise-testable languages. We also discuss some
mild generalizations of the main result, and provide bounds on
the piecewise complexity of $L\shuffle F$.
\end{abstract}


\section{Introduction}
\label{sec-intro}

Piecewise-testable languages, introduced in~\cite{simon72,simon75},
are an important variety of simple dot-depth one, hence star-free,
regular languages. As such they are closed
under boolean operations, left and right derivatives, and inverse
morphisms.

We prove in this paper that the shuffle product $L\shuffle F$ of $L$
with some finite language $F$ is piecewise-testable when $L$ is.
\\

\noindent
\emph{Some motivations.}
The question was raised by our investigations of $\FO(A^*,\subword)$,
the first-order ``logic of subwords'', and its decidable two-variable
fragment~\cite{KS-csl2016,HSZ-lics2017}.  Let us use $u\subword v$ to
denote that $u$ is a (scattered) subword, or a subsequence, of $v$.  For example,
$\texttt{simon}\subword\texttt{stimulation}$ while
$\texttt{ordering} \not\subword \texttt{wordprocessing}$.  Given a formula
$\psi(x)$ with one free variable, e.g.,
\begin{gather}
\label{ex-psi}
\tag{$\psi(x)$}
\texttt{ab}\subword x
\land \texttt{bc}\subword x
\land\texttt{ac}\not\subword x
\:,
\end{gather}
we write $\Sol(\psi)$ for its set of solutions. In this example,
$\Sol(\psi)$ is the set of all words that have $\texttt{ab}$,
$\texttt{bc}$, but not $\texttt{ac}$, among their subwords. If we
assume that the alphabet under consideration is
$A=\{\tta,\ttb,\ttc\}$, then
$\Sol(\psi)$ is
 the language described via
$\ttc^*\ttb^+\ttc(\ttb+\ttc)^*\tta^+\ttb(\tta+\ttb)^*$,
a simple regular expression. It is shown in~\cite{KS-csl2016,HSZ-lics2017} how to
compute such solutions automatically.  Let us extend the framework
with the predicate $\subword_1$, defined via
\[
u\subword_1 v \iff u\subword v\land |u| = |v| - 1,
\]
where $|u|$ is the length of $u$, so that $\subword$ is the reflexive transitive
closure of $\subword_1$. Now an $\FO^2(A^*,\subword,\subword_1)$
formula of the form
\begin{gather}
\tag{$\phi(x)$}
 \exists y: y\subword_1 x\land \psi(y)
\end{gather}
has $\Sol(\phi)= \Sol(\psi)\shuffle A$ as set of solutions.  This is
because $L\shuffle A$ is the union of all $u\shuffle a$ for $u\in L$
and $a\in A$ , and $u\shuffle a$ is the set of all words that can be
obtained by inserting the letter $a\in A$ somewhere in $u$.
Such
equalities provide an effective quantifier-elimination procedure for
(a fragment of) the logic.  Extending the complexity analysis
from~\cite{KS-csl2016} requires proving that $\Sol(\phi)$ is
piecewise-testable when $\Sol(\psi)$ is. This
will be a consequence of the main result in this paper.  \\

\noindent
\emph{Through the mirror automaton.}
It took us some time to
find a simple proof that $L\shuffle A$ is piecewise-testable when
$L$ is. In particular, starting from any of the
well-known characterizations of piecewise-testable languages (see
Definition~\ref{def-PT-multiple} below) did not take us very far.
Neither could we use the approach developed for star-free languages
---see~\cite[Coro.~3.3]{castiglione2012}--- since piecewise-testable
languages are not closed under bounded shuffle.
We eventually found a simple proof based on a classic but little-used
characterization: \emph{a regular language $L$ is
piecewise-testable if, and only if, $L$ and its mirror image
$\mirror{L}$ are $\calR$-trivial, that is, iff
 the minimal DFAs for $L$ and for $\mirror{L}$ are
both acyclic}. This characterization is not explicitly mentioned in the
main references on piecewise-testable languages,
be they  classic (e.g.,~\cite{sakarovitch83}) or
recent (e.g.,~\cite{masopust2017}). 
As far as we know, it was first given explicitly by
Brzozowski~\cite{brzozowski76b}. Beyond that, we only saw it
in~\cite{schwentick2001,klima2012b} (and derived works).
\\

\noindent
\emph{Outline of the paper.}  In section~\ref{sec-basics} we recall
the necessary notions on automata, languages, piecewise-testability,
etc., state our main result and discuss extensions.  In
Section~\ref{sec-main} we prove the main technical result: the class
of $\calR$-trivial regular languages is closed under interpolation
products with finite languages. The proof is by inspecting the
(nondeterministic) shuffle automaton and checking that the standard
determinization procedure yields an acyclic automaton. 
In Section~\ref{sec-complexity} we provide bounds on the piecewise complexity
of some shuffle languages. 
In the
conclusion, we list some questions raised by this work.



\section{Basics}
\label{sec-basics}

\noindent
\emph{Finite automata.}
We consider languages over a fixed finite alphabet $A=\{a,b,\ldots\}$
and finite automata (NFAs) of the form $\calA=(Q,A,{\cdot},I,F)$ where
``$\cdot$'' denotes the transition function.  For $p\in Q$ and $a\in
A$, $p\concatdot a$ is a subset of $Q$. 
The transition function is extended to sets of
states $S\subseteq Q$ via $S\concatdot a=\bigcup_{p\in S}p\concatdot a$ and to
words by $S\cdot\epsilon = S$ and $S\concatdot (au)=(S\concatdot a)\concatdot u$. We
often write $p\step{u}q$ rather than $q\in (p\concatdot u)$.  The language
recognized by $\calA$ is $L(\calA)\eqdef \{u\in A^* ~|~ (I\cdot u)\cap
F\neq\emptyset\}$.

$\calA$ is deterministic (is a DFA) if $|I|\leq 1$ and $|p\cdot a|\leq 1$ for all
$p$ and $a$.  It is complete if $|I|\geq 1$ and $|p\cdot a|\geq 1$ for
all $p$ and $a$.

The transition function induces a quasi-ordering on the states of
$\calA$: $p\leq_\calA q$ if there is a word $u$ such that
$p\step{u}q$, i.e., when $q$ can be reached from $p$ in the directed
graph underlying $\calA$.  The quasi-ordering is a partial ordering if
$\calA$ is acyclic, i.e., $p\step{u}q\step{v}p$ implies $p=q$; or in
other words, when the only loops in $\calA$ are self-loops. It is well
known that the $\calR$-trivial languages are exactly the languages
accepted by (deterministic) acyclic automata~\cite{brzozowski80b}.
Regarding self-loops, we
say that $p$ is $a$-stable when $p\cdot a=\{p\}$, and that it is
$B$-stable, where $B\subseteq A$ is some subalphabet, if it is
$a$-stable for each $a\in B$.  \\

\noindent
\emph{Subwords and piecewise-testable languages.}
We write $u\subword v$ when $u$ is a (scattered) subword of $v$,
i.e., can be obtained from $v$ by removing some of its letters
(possibly none, possibly all).  A word $u=a_1a_2\cdots a_n$ generates
a principal filter in $(A^*,\subword)$. This is the language $L_u
=\{v~|~u\subword v\}$, also denoted by the regular expression
$A^*a_1A^*a_2\ldots A^*a_nA^*$. The example in the introduction
has $\Sol(\psi)=
L_{\tta\ttb}\cap L_{\ttb\ttc}\cap (A^*\setminus L_{\tta\ttc})$.

For $k\in\Nat$, we write $u\sim_k v$ when $u$ and $v$ have the same
subwords of length at most $k$~\cite{simon72}. This equivalence is
called \emph{Simon's congruence} since $u\sim_k v$ implies $xuy\sim_k
xvy$ for all $x,y\in A^*$. Furthermore, $\sim_k$ partitions $A^*$ in a
finite number of equivalence classes.

\begin{definition}[Piecewise-testable languages]
\label{def-PT-multiple}
A language $L\subseteq A^*$ is piecewise-testable
if it satisfies one of the equivalent following
properties:\footnote{The last four characterizations 
refer to notions that we do not redefine in this article because we do not
use them. See references for details.}
\begin{itemize}
\item $L$ is a finite boolean combination of principal filters,
\item $L$ is a union $[u_1]_k\cup\cdots\cup [u_\ell]_k$ of
  $\sim_k$-classes for some $k\in\Nat$,
\item $L$ can be defined by a $\calB\Sigma_1$-formula in the
  first-order logic over words~\cite{DGK-ijfcs08},
\item the syntactic monoid  of $L$ is finite and $\calJ$-trivial
  (Simon's theorem)~\cite{simon72},
\item the minimal automaton for $L$ is finite, acyclic, and satisfies
  the UMS property~\cite{simon75,stern85},
\item the minimal automaton for $L$ is finite, acyclic, and
locally confluent~\cite{klima2013}.
\end{itemize}
\end{definition}
The piecewise-testable languages over some $A$ form a variety and we
mentioned the associated closure properties in our introduction.  Note
that piecewise-testable languages are not closed under alphabetic
morphisms, concatenations, or star-closures.
\\

\noindent
\emph{Shuffling languages.}
In this note we focus on the shuffle product of words and languages, and more
generally on their parameterized infiltration product. When
$C\subseteq A$ is a subalphabet and $u,v$ are two words, we let
$u\uparrow_C v$ denote the language of all words that are obtained by
shuffling $u$ and $v$ \emph{with possible sharing of letters from
$C$}. This is better defined via a notation for extracting subwords:
for a word $u=a_1a_2\cdots a_n$ of length $n$ and a subset
$K=\{i_1,\ldots,i_r\}\subseteq\{1,\ldots,n\}$ of positions in $u$
where $i_1<i_2<\cdots<i_r$, we write $u_K$ for the subword
$a_{i_1}a_{i_2}\cdots a_{i_r}$ of $u$.	Then we let
\[
x\in u\uparrow_C v
\iff \left\{
\begin{array}{l}
\exists K,K':
K\cup K'= \{1,2,\ldots,|x|\}, \\
x_K=u, x_{K'}=v, \text{ and }  x_{K\cap K'}\in C^*.
\end{array}
\right.
\]
The operation is lifted from words to languages in the standard way
via $L\uparrow_C L'=\bigcup_{u\in L}\bigcup_{v\in L'}u\uparrow_C v$.
This generalizes shuffle products and the interpolation products
$L\uparrow L'$ from~\cite{pin83,sakarovitch83} since $L\shuffle L' =
L\uparrow_\emptyset L'$ and $L\uparrow L' = L\uparrow_A L'$. Note that
$L\uparrow_C L'\subseteq L\uparrow_{C'}L'$ when $C\subseteq C'$. 
Also note that $L\uparrow_C L'=L\shuffle L'$ when
$L$ or $L'$ is subword-closed.
 A
\emph{shuffle ideal} is any language of the form $L\shuffle A^*$. It
is well-known that shuffle ideals are finite unions of principal
filters~\cite{haines69,heam2002} hence they are piecewise-testable.

\begin{theorem}[Main result]
\label{thm-main}
If $L$ is regular and $X$-trivial (where $X$ can be $\calR$, $\calL$,
or $\calJ$) then $L\uparrow_C L'$ is regular and $X$-trivial when $L'$
is finite, or cofinite, or is a shuffle ideal.
\end{theorem}
Let us first note that, since $A$ is finite, Theorem~\ref{thm-main}
answers the question about $L\shuffle A$ raised in our introduction.
A proof of the Theorem is given in the next section after a
few observations that we now make.

Let us mention a few directions in which our main result cannot be
extended:
\begin{itemize}
\item The shuffle of two piecewise-testable languages is
  star-free~\cite[Theorem~4.4]{castiglione2012} but is not always
  piecewise-testable: for example $a^*\shuffle
  ab^*$, being $a(a+b)^*$, is not piecewise-testable while $a^*$ and $ab^*$ are.

\item The concatenation $L{\cdot} F$ of a piecewise-testable $L$ and a
  finite $F$ is not always piecewise-testable: $(a+b)^*$ is
  piecewise-testable but $(a+b)^*a$ is not. Note that $L\concatdot F$ is
  included in $L\shuffle F$ that we claim is piecewise-testable.

\item The scattered residual $L \dashrightarrow u$ of a
  piecewise-testable $L$ by some word $u$ is not always
  piecewise-testable.  For example $ac(a+b)^* \dashrightarrow
  c=a(a+b)^*$.  (Recall that $w \dashrightarrow u$ is the set of all
  words $v$ such that $w\in u\shuffle v$, obtained by removing the
  subword $v$ somewhere along $w$~\cite{kari94}.)
\end{itemize}

Finally, there are some (admittedly degenerate) situations that are
not covered by Theorem~\ref{thm-main} and where the shuffle of
two piecewise-testable languages is piecewise-testable.
\begin{proposition}
If $L_1,\ldots,L_m\subseteq A^*$ are piecewise-testable then
$L_1\shuffle \cdots \shuffle L_m$ is piecewise-testable in any of the following
cases:
\begin{itemize}
\item the $L_i$'s are all complements of shuffle ideals, i.e., they are
subword-closed;
\item their subalphabets are pairwise disjoint.
\end{itemize}
\end{proposition}
The first claim is easy to see since the shuffle of subword-closed
languages is subword-closed, and the second claim\footnote{Already
given in the long version of~\cite{masopust2016b}.}
is a consequence of the following Lemma.
\begin{lemma}[{\protect See also~\cite[Lemma~6]{esik98}}]
\label{lem-disjoint-alpha}
Let $\frakF$ be a family of languages over $A$ that is closed under
intersections and inverse morphisms.  If $L_1,L_2\in\frakF$ use
disjoint subalphabets, then $L_1\shuffle L_2$ is in $\frakF$ too.
\end{lemma}
\begin{proof}
Write $e_B: A^*\to A^*$ for the erasing morphism that replaces all
letters from some subalphabet $B$ with $\epsilon$ and leaves other
letters unchanged.  Assuming $L_1\subseteq A_1^*$ and $L_2\subseteq
A_2^*$, with furthermore $A_1\cap A_2=\emptyset$, one has
\begin{gather*}
L_1\shuffle L_2 = (L_1\shuffle A_2^*)\cap (L_2\shuffle A_1^*)
= e_{A_2}^{-1}(L_1) \cap e_{A_1}^{-1}(L_2)
\:.
\end{gather*}
The last equality shows that $L_1\shuffle L_2$ is in $\frakF$.
\end{proof}



\section{Shuffling acyclic automata}
\label{sec-main}

In this section we first prove Proposition~\ref{prop-main} by inspecting the shuffling
of automata.
\begin{proposition}
\label{prop-main}
If $L\subseteq A^*$ is regular and $\calR$-trivial then $L\uparrow_C
w$ is too, for any $w\in A^*$ and $C\subseteq A$.
\end{proposition}
Let $\calA=(Q,A,{\cdot},i,F)$ be an acyclic complete deterministic
automaton for $L$, and let $w=z_1\cdots z_m\in A^*$ be the word under
consideration. When building the shuffle automaton for $L\uparrow_C
w$, it is more convenient to consider the smallest automaton for $w$,
deterministic but not complete.	 Formally, we let
$\calB=(Q',A,{\Bcdot},i',F')$ given by $Q' = Q\times\{0,1,\ldots,m\}$,
$i'=(i,0)$, $F'=F\times \{m\}$, and a transition table given by
\begin{equation}
\label {eq-delta-B}
(p,k)\Bcdot a=
\bigl\{
\;\;
(p\cdot a,k)
,\;\;
\obracew{(p,k+1)}{\text{if $a{=}z_{k+1}$}}
,
\obracew{(p\cdot a,k+1)}{\text{if furthermore $a\in C$}}
\bigr\}.
\end{equation}
This is a standard construction: $\calB$ is nondeterministic in
general, and it is easy to see that it accepts exactly $L\uparrow_C
w$.

Observe that $\calB$ too is acyclic: by Eq.~\eqref{eq-delta-B}, for
any transition $(p,k)\step{a}(q,\ell)$ one has $p\leq_\calA q$ and
$k\leq\ell$ and this extends to any path $(p,k)\step{u}(q,\ell)$ by
transitivity. Thus $\leq_\calB$ is included in the Cartesian product
of two partial orderings.

From $\calB=(Q',A,{\Bcdot},i,F')$ we derive a powerset automaton
$\calP=(\bmQ,A,{\Pcdot},\bmi,\bmF)$  in the standard way, i.e.,
$\bmQ=2^{Q'}=\{S~|~S\subseteq Q'\}$, $\bmi=\{i'\}$,
$\bmF=\{S\in\bmQ~|~S\cap F'\neq \emptyset\}$ and $S\Pcdot
a=\{S\Bcdot a\}$. It is well known that $\calP$ is
deterministic, complete, and accepts exactly the language accepted by
$\calB$, i.e., $L\uparrow_C w$.
\begin{lemma}
\label{lem-P-acyclic}
$\calP$ is acyclic.
\end{lemma}
\begin{proof}
Let $S_0\step{a_1}S_1\step{a_2}S_2\cdots\step{a_n}S_n=S_0$ be a
non-empty cycle in $\calP$ and write $S=\bigcup_{i=0}^n S_i$ and
$B=\{a_1,\ldots,a_n\}$ for the set of states (resp., set of letters)
appearing along the cycle.

We first claim that for any $(p,k)\in S_n$, $p$ is $B$-stable in $\calA$,
which mean that $p\cdot a_i=p$ for $i=1,\ldots,n$. We prove this by induction
on $\leq_\calB$: so consider an arbitrary $(p,k)\in S_n$ and assume
that $p'$ is $B$-stable whenever there is some $(p',k')\in S_n$ with $(p',k')<_\calB(p,k)$.
Since $S_0\step{a_1}S_1\cdots\step{a_n}S_n$ and $(p,k)\in S_n$,
$\calB$ has a sequence of transitions
\[
(p_0,\ell_0)
\:\step{a_1}\: (p_1,\ell_1)
\:\step{a_2}\: (p_2,\ell_2)
\cdots
\:\step{a_n}\: (p_n,\ell_n) = (p,k)
\]
with $(p_i,\ell_i) \in S_i$ for all $i=1, \ldots, n$. Thus $p_0
\leq_\calA p_1 \cdots \leq_\calA p_n=p$ and $\ell_0 \leq \ell_1 \cdots
\leq \ell_n=k$.  If $p_0\neq p$, then $p_0 = p_1 = \ldots = p_{i-1}
\neq p_i\leq_\calA p_n$ for some $i$. Given $(p_{i-1},\ell_{i-1})
\step{a_i} (p_i,\ell_i)$ and $p_{i-1}\neq p_i$,
Eq.~\eqref{eq-delta-B} requires that $p_{i-1}\cdot {a_i}=p_{i}$ in
$\calA$, hence $p_{i-1}$ is not $B$-stable, but this contradicts the
induction hypothesis since $p_{i-1}=p_0$, $(p_0,\ell_0)$ belongs to $S_n$, and
$(p_0,\ell_0)<_\calB(p,k)$. Thus $p_0=p_1=\cdots=p_n=p$. If $\ell_0<\ell_n$, the
induction hypothesis applies and states that $p_0$ is $B$-stable. If
$\ell_0 = \ell_1 = \cdots = \ell_n$, then Eq.~\eqref{eq-delta-B}
requires that $p_{i-1} \cdot a_i=p_i$ for all $i<n$, which proves the
claim.

Since we can change the origin of the cycle, we conclude that $p$ is
$B$-stable in $\calA$ for any $(p,k)$ in $S$, not just in $S_n$. If
$p$ is $B$-stable, then $(p,k)\Bcdot a_i\ni (p,k)$ by
Eq.~\eqref{eq-delta-B}. Thus $S_{i-1}\Pcdot a_i \supseteq S_{i-1}$ for
all $i=1,\ldots,n$. This entails $S_0 \subseteq S_1 \subseteq \cdots
\subseteq S_n=S_0$ and then $S_0 = S_1 = \ldots = S_n$. We have proved
that all cycles in $\calP$ are self-loops, hence $\calP$ is acyclic as
claimed.

This entails that $L\uparrow_C w$, the language recognized by $\calP$,
is $\calR$-trivial and concludes the proof of Proposition~\ref{prop-main}.
\end{proof}

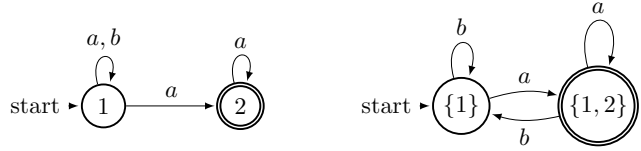
\begin{figure}[htbp]
\centering
\scalebox{0.9}{
\begin{tikzpicture}[>=latex,shorten >=1pt,auto,node distance=2cm,
scale=0.40,every state/.style={inner sep=2pt,minimum size=18pt,thick}]
\node[initial,state]     (q1)                          {$1$};
\node[accepting,state]   (q2)    [right of=q1]         {$2$};

\path[->] (q1) edge [loop above] node{$a,b$} (q1);
\path[->] (q1) edge node[above] {$a$} (q2);
\path[->] (q2) edge [loop above] node{$a$} (q2);

\node[initial,state]     (S1)    [right=7em of q2]     {$\{1\}$};
\node[accepting,state]   (S2)    [right of=S1]         {$\{1,2\}$};
\path[->] (S1) edge [bend left=15] node[above] {$a$} (S2);
\path[->] (S2) edge [bend left=15] node[below] {$b$} (S1);
\path[->] (S1) edge [loop above]   node{$b$} (S1);
\path[->] (S2) edge [loop above]   node{$a$} (S2);
\end{tikzpicture}
}
\caption{NFA for $a^* \shuffle b^*a$ and associated powerset DFA.}
\label{fig-power-automaton}
\end{figure}

\begin{remark}
\label{rem-nfa-acyclic}
\Cref{lem-P-acyclic} needs a proof because determinizing an acyclic
NFA does not always yield an acyclic
DFA.\footnote{Indeed nondeterministic and deterministic acyclic
automata have different expressive powers, see~\cite{schwentick2001}.}  For
example, the NFA obtained by shuffling DFAs for $a^*$ and for $b^*a$
is acyclic (see left of Fig.~\ref{fig-power-automaton}).  However, its
powerset automaton and the minimal DFA are not (see right of the
figure). Indeed, $a^*\shuffle b^*a=(a+b)^*a$ is not $\calR$-trivial.
\end{remark}

With Proposition~\ref{prop-main} it is easy to prove our main result.
\begin{proof}[Proof of Theorem~\ref{thm-main}]
We first assume that $L$ is $\calR$-trivial
and  consider several cases for $L'$:
\begin{itemize}
\item
If $L'$ is finite, we use distributivity of shuffle over unions:
$L\uparrow_C L'$ is $\calR$-trivial since it is a finite union
$\bigcup_{w\in L} L\uparrow_C w$ of $\calR$-trivial languages.
\item
If $L'$ is a shuffle ideal, i.e., if $L'=L'\shuffle A^*=L'\uparrow_C
A^*$, then $L\uparrow_C L'$ is a shuffle ideal too in view of
\[
L\uparrow_C L'
=
L\uparrow_C (L'\uparrow_C A^*)
=
(L\uparrow_C L')\uparrow_C A^*
\:.
\]
Recall now that shuffle ideals are always $\calR$-trivial.
\item
If $L'$ is cofinite, it is the union of a finite language and a
shuffle ideal, so this case reduces to the previous two cases by
distributing shuffle over union.
\end{itemize}
Once the result is proved for $X=\calR$, it extends to $X=\calL$ by
mirroring since $L$ is $\calL$-trivial if, and only if, its mirror
$\mirror{L}$ is $\calR$-trivial, and since $\mirror{(L\uparrow_C L')}=
\mirror{L}\uparrow_C\mirror{L'}$.

Finally, it extends to $X=\calJ$ since a finite monoid is
$\calJ$-trivial if, and only if, it is both $\calR$- and $\calL$-trivial.
\end{proof}

\begin{remark}
Masopust and Thomazo extended the UMS criterion to \emph{nondeterministic}
automata.  They showed that $L$ is piecewise-testable if it is
recognized by a complete acyclic NFA with the UMS
property~\cite[Thm.~25]{masopust2017}. The NFA that one obtains by
shuffling minimal DFAs for $L$ and $w$ is indeed acyclic and
complete. However it does not satisfy the UMS property in general
(already with $a^*\shuffle a$) so this additional characterization of
piecewise-testable language does not  directly entail our main result.
\end{remark}


\section{The question of piecewise complexity}
\label{sec-complexity}

We write $h_A(L)$ for the \emph{piecewise complexity} of $L$, defined
as the smallest $k$ such that $L$ is $k$-$\PT$, i.e., can be written as
a union $L=[u_1]_k \cup \cdots \cup [u_r]_k$ of $\sim_k$-classes over
$A^*$. We let $h_A(L)=\infty$ when $L$ is not piecewise-testable. For
notational convenience, we usually write $h(L)$ when the alphabet is
understood\footnote{The only situation where $A$ is relevant happens
for $h_A(A^*)=0<h_{A'}(A^*)=1$ when $A\subsetneq A'$.}  and
write $h(u)$ for $h(\{u\})$ when $L=\{u\}$ is a singleton.

It was argued in~\cite{KS-csl2016} that $h(L)$ is an important, robust
and useful, descriptive complexity measure for $\PT$ languages. In
this light, a natural question is to provide upper-bounds on
$h(L\shuffle L')$ as a function of $h(L)$ and $h(L')$.  
Computing or bounding $h(L)$ has received little attention
until~\cite{KS-csl2016}, and the available toolset for these questions
is still primitive. In this section we provide some preliminary
answers for $L\shuffle L'$ and slightly enrich the available toolset.
\\

Before looking at simpler situations, let us
note that, in general, the piecewise-complexity of $L\shuffle w$ can
be much higher than $h(L)$ and $h(w)$.
\begin{proposition}[Complexity blowup]
\label{prop-blowup}
One cannot bound $h(L\shuffle w)$ with a
polynomial of $h(L)+h(w)$, even if we require $h(L)= 0$.
(NB: this statement assumes unbounded alphabets.)
\end{proposition}
\begin{proof}
Pick some $\lambda\in\Nat$ and let $U_n$ be a word over a $n$-letter
alphabet $A_n=\{a_1,\ldots,a_n\}$, given by $U_0=\epsilon$ and
$U_{i+1}=(U_i a_{i+1})^\lambda U_i$.  It is known that
$h(U_n)=n\lambda+1$~\cite[Prop.~3.1]{KS-csl2016}.
On the other hand $h(A_n^*\shuffle
U_n)=h(L_{U_n})=|U_n|=(\lambda+1)^n-1$ since, for any word $u$,
$h(L_u)=|u|$~\cite[Prop.~4.1]{KS-csl2016}.
\end{proof}

\subsection{Simple shuffles}

\begin{proposition}
\label{prop-h-disjoint-alpha}
Assume that $L_1$ and $L_2$ are two non-empty piecewise-testable
languages on disjoint alphabets. Then $h(L_1\shuffle L_2)=\max
(h(L_1),h(L_2))$.
\end{proposition}
\begin{proof}
Since $k$-$\PT$ languages form a variety~\cite[Lemma~2.3]{therien81},
\Cref{lem-disjoint-alpha} applies and yields $h(L_1\shuffle
L_2)\leq \max(h(L_1),h(L_2))$.

To see that $h(L_1\shuffle L_2)\geq h(L_1)$, we  write $k=h(L_1\shuffle
L_2)$ and show that $L_1$ and $L_2$ are closed under $\sim_k$:
Pick any word $u\in L_1$ and any $u'\in A_1^*$ with $u\sim_k
u'$. Since $L_2$ is not empty, there is some $v\in L_2$ and we obtain
$uv\in L_1\shuffle L_2$, and also $u'v\in L_1\shuffle L_2$ since
$uv\sim_k u'v$. Necessarily $u'\in L_1$ since $L_1$ and $L_2$ have
disjoint alphabets. Hence $L_1$
is closed under $\sim_k$, i.e., $h(L_1)\leq k$. The same
reasoning applies to $L_2$.
\end{proof}

\begin{proposition}
\label{prop-h-shuf-filters}
Assume that $L_u$ and $L_v$ are two principal filters.
Then $h(L_u\shuffle L_v)\leq h(L_u)+h(L_v)$.
\end{proposition}
\begin{proof}
Recall that $h(L_u)=|u|$ as noted above. We then observe
that $L_u\shuffle L_v=\bigcup_{w\in u\shuffle v}L_w$ and that
$|w|=|u|+|v|$ for all $w\in u\shuffle v$.
\end{proof}
The upper bound in \Cref{prop-h-shuf-filters} can be reached, an easy
example being $h(L_{a^n}\shuffle L_{a^m})=h(L_{a^{n+m}})=n+m$.  The
inequality can also be strict, as exemplified by
\Cref{prop-h-disjoint-alpha}.

\subsection{Shuffling finitely many words}

Finite languages are piecewise-testable and closed under shuffle
products. Their piecewise complexity reduces to the case of individual
words
in view of the following (from~\cite{KS-csl2016}):
\begin{gather}
\label{eq-h-F}
h(F) = \max_{u\in F} h(u)
\quad
\text{when $F$ is finite}.
\end{gather}
\begin{lemma}
\label{lem-h-shuffle-words}
$h(u_1\shuffle u_2 \shuffle \cdots \shuffle u_m)\leq
1+\max_{a\in A}\bigl(|u_1|_a+\cdots+|u_m|_a\bigr)$.
\end{lemma}
\begin{proof}
Assume $A=\{a_1,\ldots,a_n\}$ and define $\ell_1,\ell_2,\ldots,\ell_n$
via $\ell_j=|u_1|_{a_j}+\cdots+|u_m|_{a_j}$.  From
\[
u_1 \shuffle \cdots \shuffle u_m \:\subseteq\:
a_1^{\ell_1}\shuffle \cdots \shuffle a_n^{\ell_n}
\:,
\]
we deduce
\begin{align*}
h(u_1\shuffle \cdots \shuffle u_m)&\leq
h\bigl(a_1^{\ell_1}\shuffle \cdots \shuffle a_n^{\ell_n}\bigr)
\\
\shortintertext{by Eq.~\eqref{eq-h-F}}
&= \max
\bigl( h(a_1^{\ell_1}),\ldots,h(a_n^{\ell_n})\bigr)
\\
\shortintertext{by Prop.~\ref{prop-h-disjoint-alpha}}
&=
\max(1+\ell_1,\ldots,1+\ell_n)\:.\qedhere
\end{align*}
\end{proof}
We may now bound $h(u_1\shuffle u_2\shuffle  \cdots)$ as a function of
$h(u_1),h(u_2),\ldots$.
\begin{theorem}[Upper bound for shuffles of words]
\label{thm-h-shuffle-words}
Assume $|A|=n$.\\
(1) $h(u_1\shuffle u_2 \shuffle \cdots \shuffle u_m)$
is in $O\bigl(\bigl[\sum_{i=1}^m h(u_i)\bigr]^n\bigr)$.
\\
(2)
This upper bound is tight: for every $\lambda\in\Nat$, there exists words
$u_1,\ldots,u_m$ with fixed $m=n$ and such that $h(u_1\shuffle \cdots \shuffle
u_m)=(\lambda+1)^n$ and $h(u_1)+\cdots+h(u_m)=n^2\lambda+n$.
\end{theorem}
\begin{proof}
(1) 
By \Cref{lem-h-shuffle-words},
\begin{align*}
 & h(u_1\shuffle u_2 \shuffle \cdots \shuffle u_m)-1
\\
\leq &
\max_{a\in A}\bigl(|u_1|_a+\cdots+|u_m|_a\bigr)
\leq
\sum_{i=1}^m |u_i|
\:.
\end{align*}
On the other hand, \cite[Prop.~3.8]{KS-csl2016}  showed that
\[
|u|< \left(\frac{h(u)}{|A|}+2\right)^{|A|} \text{ for any word $u\in A^*$.} 
\]
Thus, for fixed $A$, $|u|$ is  $O(h(u)^{|A|})$  and
$\sum_i |u_i|$ is  $O\bigl(\bigl[\sum_i h(u_i)\bigr]^{|A|}\bigr)$, which
establishes the upper bound claim.

\noindent
(2) We consider $U_n$ as defined in the proof of Proposition~\ref{prop-blowup} and, for
$j=1,\ldots,m$, let $u_{j}$ be $r^j(U_n)$ where $r:A^*\to A^*$ is the
circular renaming that replaces each $a_i$ by $a_{i+1}$ (counting
modulo $n$). Write $\ell$ for $|U_n|$, i.e., $\ell=(\lambda+1)^n-1$.  We
saw that $h(u_{j})=h(U_n)=n\lambda+1$ so, fixing $m=n$, $\sum_{i=1}^m
h(u_i)=n^2\lambda+n$ as claimed. Let $L = u_{1}\shuffle 
u_{2}\shuffle \cdots\shuffle u_{n}$. There remains to prove that
$h(L)=(\lambda+1)^n=\ell+1$.

We first observe that, for any letter $a_j$, $|u_{1}|_{a_j} + \cdots +
|u_{n}|_{a_j} = \ell$.  Indeed, the circular renamings ensure that
\[
|r^1(u)|_{a_j}+\cdots+|r^n(u)|_{a_j}=
|u|_{a_{j-1}}+\cdots+|u|_{a_{j-n}}=|u|
\]
for any word $u\in A^*$.  We then obtain $h(L)\leq \ell+1$ by \Cref{lem-h-shuffle-words}.

There remains to show $h(L)>\ell$.  For this, we observe that, for any
$i=1,\ldots,\ell$, the $i$-th letters $u_{1}[i],\ldots,u_{n}[i]$ form
a permutation of $\{a_1,\ldots,a_n\}$. Thus we can obtain
$(a_1a_2\cdots a_n)^\ell$ by shuffling $u_{1},\ldots,u_{n}$, i.e.,
$(a_1a_2\cdots a_n)^\ell\in L$.  However $(a_1a_2\cdots
a_n)^{\ell}a_1$ is not in $L$ (it is too long) and $(a_1a_2\cdots
a_n)^{\ell}a_1 \sim_{\ell} (a_1a_2\cdots a_n)^\ell$ (both words
contain all possible subwords of length $\leq\ell$). Thus $L$ is not
closed under $\sim_\ell$, which concludes the proof.
\end{proof}

\subsection{A general upper bound?}

As yet we do not have a good upper bound in the general case.  

Recall
that the \emph{depth} of a complete DFA is the maximal length of an acyclic path
from some initial to some reachable state. When $L$ is regular, we write
$\depth(L)$ for the depth of the canonical DFA for $L$.  Since
$h(L)\leq\depth(L)$ holds for all $\PT$ languages~\cite{klima2013}, one
could try to  bound
$\depth(L\shuffle w)$ in terms of $\depth(L)$ and $w$.
This does not seem very promising: First, for $L$ fixed, $\depth(L\shuffle w)$ cannot be bounded in
$O(|w|)$.  Furthermore, $\depth(L)$ can be much larger than $h(L)$: if
$L$ is $k$-$\PT$ and $|A|=n$ then the depth of the minimal DFA for $L$
can be as large as $\binom{k+n}{k}-1$~\cite[Thm.~31]{masopust2017}. Finally, this approach would only
provide very large upper bounds, far above what we observe in
experiments.



\section{Conclusion}
\label{sec-concl}

We proved that $L\shuffle w$ is piecewise-testable when $L$ is (and
when $w$ is a word), relying on a little-used characterization of
piecewise-testable languages. This is part of a more general research
agenda: identify constructions that produce piecewise-testable
languages and compute piecewise complexity modularly. In this
direction, an interesting open problem is to identify sufficient
conditions that guarantee that a Kleene star $L^*$, or a concatenation
$L\concatdot L'$, is piecewise-testable. It is surprising that such
questions seem easier for shuffle product than for concatenation.


\paragraph{Acknowledgments}
We thank Prateek Karandikar, Amaldev Manuel, Laurent Doyen, Stefan
Schwoon and the anonymous reviewers for their comments and suggestions.



\bibliographystyle{alpha}
\bibliography{used}

\end{document}